\documentclass[12pt]{article}

\usepackage[utf8]{inputenc}
\usepackage[T1]{fontenc}
\usepackage[left=0.75in,right=0.75in,top=0.95in]{geometry}
\usepackage{mathptmx}
\usepackage{microtype}
\usepackage{enumerate}
\usepackage{amsmath,amssymb,amsthm}
\usepackage{graphicx, xcolor}
\usepackage{hyperref, cleveref, xurl}
\usepackage[style=alphabetic]{biblatex}
\usepackage{bbm, enumitem}

\hypersetup{
    colorlinks=true,
    citecolor=green!50!black,
    linkcolor=olive
    }
\newcommand{\X}{\mathcal X_k}
\newcommand{\Y}{\mathcal Y_n}
\newcommand{\Z}{\mathcal Z_{k,n}}
\newcommand{\iid}{\stackrel{\textnormal{iid}}{\sim}}
\newcommand{\R}{\mathbb R}
\renewcommand{\P}{\mathbb P}
\newcommand{\E}{\mathbb E}
\newcommand{\Var}{\operatorname{Var}}
\newcommand{\Hnull}{\textnormal{H}_0}
\newcommand{\Hone}{\textnormal{H}_{1,\mathcal P}}

\newtheorem{theorem}{Theorem}
\newtheorem{corollary}{Corollary}
\theoremstyle{remark}
\newtheorem{remark}{Remark}
\theoremstyle{definition}

\newtheorem*{example*}{Example}

\title{Exact and efficient multivariate two-sample tests through adaptive linear multi-rank statistics}
\author{Dan Daniel Erdmann-Pham \\[2mm] \emph{Department of Statistics, Stanford University} \\[2mm] \texttt{erdpham@stanford.edu}}
\date{}

\begin{document}

\maketitle

\begin{abstract}
    So-called linear rank statistics provide a means for distribution-free (even in finite samples), yet highly flexible, two-sample testing in the setting of 
    univariate random variables. Their flexibility derives from a choice of weights that can be adapted to any given (simple) alternative hypothesis to achieve
    efficiency in case of correct specification of said alternative, while their non-parametric nature guarantees well-calibrated $p$-values even under
    misspecification. By drawing connections to (generalized) maximum likelihood estimation, and expanding on recent work on ranks in multiple dimensions, we extend
    linear rank statistics both to multivariate random variables and composite alternatives. Doing so yields non-parametric, multivariate two-sample tests that mirror
    efficiency properties of likelihood ratio tests, while remaining robust against model misspecification and computationally tractable. We prove non-parametric versions of the classical Wald and
    score tests facilitating hypothesis testing in the asymptotic regime, and relate these generalized linear rank statistics to linear spacing statistics enabling
    exact $p$-value computations in the small to moderate sample setting. Moreover, viewing rank statistics through the lens of likelihood ratios affords applications
    beyond fully efficient two-sample testing, of which we demonstrate three: testing in the presence of nuisance alternatives, simultaneous detection of location and
    scale shifts, and $K$-sample testing.
\end{abstract}

\section*{Introduction}

The task of deciding whether two samples $\X = \{X_1, ..., X_k\} \iid F$ and $\Y = \{Y_1, ..., Y_n \}\iid G$ ($\X, \Y\subset\R^p$) have arisen from the same ($F=G$) or distinct ($F\neq G$) distributions is known as two-sample testing, and is a staple in statistical analysis \autocite[see, e.g., ][and references therein]{lehmann1975nonparametrics, thas2010comparing, weiss1960two,jurevckova2012nonparametric}, with applications ranging from genomics \cite{charmpi2015weighted} to econometrics \cite{goerg2009nonparametric} to physics \cite{aslan2005statistical} among many. When parametric assumptions on $\X$ and $\Y$ are imposed; e.g., $F,G \in \mathcal P$ for some sufficiently regular, finite-dimensional family of absolutely continuous distributions $\mathcal P = \{p_{\theta} : \theta \in \Theta \subset \R^m\}$, then a test based on the generalized likelihood ratio 
\begin{equation}
    \lambda_{\mathcal P}(\X, \Y) = \frac{ \max_{\theta_1,\theta_2 \in \Theta} p_{\theta_1}(\X) p_{\theta_2}(\Y) }{ \max_{\theta\in\Theta} p_{\theta}(\X) p_{\theta}(\Y) }
    \label{eq:likelihood_ratio}
\end{equation}
displays various optimality properties \cite{van2000asymptotic, bahadur1967optimal, zeitouni1992generalized}, and is easily computed when $\min\{k,n\}$ is large as $-2\log \lambda_{\mathcal P}(\X, \Y) \to \chi^2_m$ (which we refer to as Wilks phenomenon, following analogous behavior in the one-sample situation derived in \autocite{wilks1938large}), rendering it the typical method of choice whenever a high-confidence candidate for $\mathcal P$ is available. However, when confidence in $\mathcal P$ is low; that is, the user may expect $F,G$ to be reasonably close to (in some suitable sense), yet not exactly contained within $\mathcal P$, then care must be taken with tests based on $\lambda_{\mathcal P}$: it is known that Wilks phenomenon does not persist under model misspecification, and a limiting misspecification-dependent non-central $\chi^2$ distribution obtains instead \autocite{white1982maximum}. Therefore, $p$-values based on $\lambda_{\mathcal P}$ are, in general, not well calibrated in this situation. A second difficulty commonly encountered with any maximum likelihood approach lies in its computational complexity: although computing resources are becoming ever more elaborate, so do data sets and their associated models, often rendering the two optimizations involved in \eqref{eq:likelihood_ratio} challenging. The goal of this article is to propose an alternative to likelihood-ratio type tests that retains the favourable optimality properties of $\lambda_{\mathcal P}$, while simultaneously being computationally feasible and robust to model misspecification.

A wealth of two-sample tests addressing different aspects of this goal have been developed: in the univariate case $\X, \Y \subset \mathbb R$, statistics based on the ranks of $\X$ and $\Y$ in the pooled sample $\Z = \X \cup \Y$ are distribution-free under $\Hnull$ and therefore automatically robust against misspecification of any kind. Popular examples of such tests include the Anderson-Darling and Cram\'er-von Mises tests, as well as the Mann-Whitney-$U$, van-der-Waerden, Siegel-Tukey (equivalent to Bradley-Freund-Ansari), Mood, and Klotz tests \autocite[see, e.g.,][for overviews and detailed references]{lehmann1975nonparametrics}. While the former two are based on notions of distance between the empirical distributions functions $F_k, G_n$ of $\X$ and $\Y$, the latter fall under the umbrella of so-called linear rank statistics, using test statistics of the form
\begin{equation*}
    T_{k,n}^w = \frac{1}{N}\sum_{j=1}^n w\left( \frac{r_j}{N} \right) = (1-\alpha_{k,n}) \int w \circ H_{k,n} \ \mathrm{d}G_{n},
\end{equation*}
where $H_{k,n} = \alpha_{k,n} F_k + (1-\alpha_{k,n}) G_n$ with $N = k+n$ and $\alpha_{k,n} = k/N$, is the empirical distribution of $\Z$, $r_j$ is the rank of $Y_j$ in $\Z$, and $w: [0,1] \to \R$ is a weight function chosen by the user. The shape of $w$ dictates the alternatives against which $T_{k,n}^w$ is powerful, with the Mann-Whitney-$U$ ($w(x) = x$) and van-der-Waerden ($w = \Phi^{-1}$, where $\Phi$ is the Gaussian CDF) tests designed to detect differences in location (i.e., $G(x) = F(x-\theta)$), while Siegel-Tukey ($w(x) = | x - 1/2 |$), Mood ($w(x) = (x-1/2)^2$), and Klotz ($w = \Phi^{-2}$) seek to differentiate shifts in scale $G(x) = F(x/\theta)$. Despite their non-parametric nature, linear rank statistics often perform surprisingly well compared to their parametric counterparts: e.g., \autocite{hodges1956efficiency} showed that the Pitman efficiency against local contiguous alternatives \autocite[for relevant definitions see, e.g., ][Chapter 14]{van2000asymptotic} of Mann-Whitney's $U$ with respect to the $t$-test is lower bounded by $108/125 \approx 0.864$ over all choices of $F$. More strikingly yet, the same efficiency of van-der-Waerden's test against the $t$-test is always at least $1$, as proven in \autocite{chernoff1958asymptotic}; that is, at least asymptotically in the Pitman sense, $T_{k,n}^{\Phi^{-1}}$ is always preferable to the $t$-test. Practical implementations of linear rank tests typically rely on the asymptotic normality of $T_{k,n}^w$ to determine rejection thresholds, though \autocite{erdmann2022generalized} has shown that exact $p$-values in the finite-sample setting can be obtained for so called linear rank spacings, which are closely related and asymptotically equivalent to using linear ranks \autocite{holst1980asymptotic}.

The primary difficulty in extending results such as these to the multivariate situation $\X, \Y \subset \R^p, p>1$ lies in the absence of a canonical ordering, leading to ambiguity of the notion of ranks in higher dimensions. Various extensions have been proposed \autocite{bickel1965some, chaudhuri1996geometric, marden1999multivariate, zuo2000general, hallin2002optimal, hallin2004rank, randles1989distribution, peters1990multivariate}, though two-sample tests associated with them typically do not exhibit the same distribution-freeness that their univariate counterparts do, and can pose computational challenges. Multivariate two-sample tests that are not based on generalized ranks generally leverage either geometric graphs \autocite{friedman1979multivariate, henze1988multivariate, chen2017new, rosenbaum2005exact} or embeddings into reproducing kernel Hilbert spaces \autocite{gretton2012kernel}, of which energy-distances are a special case \autocite{baringhaus2004new}. The latter usually lack distribution-freeness (even asymptotically), necessitating generic (often conservative) tail bounds or permutation schemes, while the former are known to have $0$ Pitman efficiency in most settings \autocite{bhattacharya2019general}. Recently, through the theory of optimal transport, \autocite{hallin2020distribution} introduced a multivariate generalization of the rank map that does inherit most desirable properties of its univariate relative, which \autocite{deb2021efficiency} used (and expanded) effectively in the design of a multivariate two-sample test that is exactly distribution-free and Pitman efficient against the Hotelling $T^2$ test \autocite[the suitable multivariate analogue of the $t$-test,][]{hotelling1992generalization} in the location setting. There, the authors start from the interpretation of Mann-Whitney's $U$ as a "rank version" of the $t$-test; that is, $U = k^{-1}\sum_{j=1}^k H_{k,n}(X_j) - n^{-1}\sum_{j=1}^n H_{k,n}(Y_j)/n$ compared with $t = k^{-1}\sum_{j=1}^k X_j - n^{-1}\sum_{j=1}^k Y_j$ (in the situation where $\Var X$ and $\Var Y$ are known and identical), and proceed to similarly define a rank version of Hotelling's $T^2$. The multivariate ranks involved in this construction, however, require computation of optimal transport maps between two samples of size $N = k+n$, which generically are expensive to obtain: exact solvers may take $O(pN^2 + N^3\log^2 N)$, while even approximate solutions may necessitate $O(pN^2 + N^2\log N)$ \autocite[][p. 42, 65]{COTFNT}. With $N$ frequently ranging into the thousands or millions in, e.g., modern biology, such computation quickly grows intractable. Our first contribution consists in identifying alternative transport maps that, too, allow for efficient multivariate two-sample testing, yet can be computed in $O(pN\log N)$. Moreover, in order to achieve efficiency, \autocite{deb2021efficiency} introduce a multivariate analogue $J: \R^p \to \R^p$ of the weight function $w$ defining $T_{k,n}^w$, and show that it performs well (relative to Hotelling's $T^2$) as long as a so-called effective reference distribution resulting from the multivariate rank map and $J$ is Gaussian (mirroring van-der-Waerden's choice of $w = \Phi^{-1}$). A number of effective reference distributions other than the Gaussian one are investigated, and shown to exhibit substantial differences in relative efficiency compared to Hotelling's $T^2$, with, e.g., some displaying significant dependence on $p$, the dimension of $\X$ and $\Y$, while others do not. In the face of such heterogeneity, it is natural to ask: do some choices of reference distributions perform better for a given choice of $F$, and is there an optimal such choice? Indeed, this question extends beyond the setting of detecting location shifts, and can be asked in the more general setting of the discussion surrounding \eqref{eq:likelihood_ratio}. Given that every test has meaningful power against local alternatives only in finitely many directions \autocite[see, e.g., ][Chapter 14.6]{lehmann2005testing}, providing practitioners with a clear recipe to turn a subspace of alternatives that they suspect may be relevant for the application at hand into a test powerful along said subspace should be valuable. By tying the notion of effective reference distributions and weight functions to $\lambda_{\mathcal P}$, we hope to give some guidance on how such a recipe might look like.

\section*{Univariate linear rank statistics as local likelihood ratios}

Pitman efficiencies measure relative performances of two tests against local alternatives; that is, against alternatives $G_n$ converging to $F$ at rate $N^{-1/2}$ in some suitable sense. To fix ideas, assume for now that $\mathcal P$ is univariate, twice continuously differentiable, $0\in \Theta = (a,b)$, $F = P_0$, and $G= P_{\theta/\sqrt{N}}$ for some $\theta\in\R$. This simple situation serves to illustrate main ideas, with extensions to more general settings and a more rigorous treatment following in later sections. In this context, the log-likelihood of $\Y$ under $G$ reads
\begin{equation*}
    \log p_{\theta/\sqrt{N}}(\Y) \stackrel{G}{\approx} \frac{\theta}{\sqrt{N}} z_0(\Y) - \frac{\theta^2}{2N} \partial_{\theta} z_0(\Y) + o_{G}(N^{-1}),
\end{equation*}
where $z_0$ is the score associated with $\mathcal P$ at $\theta = 0$, the notation $f(\Y)$ is used to mean $\sum_{Y\in\Y} f(Y)$ (occasionally it may represent $\prod_{Y\in\Y} f(Y)$, which will be clear from context; e.g., when $f$ represents a density), and the approximation is valid as long as $N$ is sufficiently large and $\alpha_{k,n} \to \alpha > 0$. Heuristically,
\begin{equation*}
    H_{k,n}(Y) \stackrel{G}{\approx} P_0(Y) + (1-\alpha_{k,n})\frac{\theta}{\sqrt{N}} \partial_{\theta} P_0(Y),
\end{equation*}
and so
\begin{equation*}
    \theta T_{k,n}^{z_0\circ P_0^{-1}} = \frac{\theta}{\sqrt{N}} z_0 \circ P_0^{-1}\circ H_{k,n}(\Y) \stackrel{G}{\approx} \frac{\theta}{\sqrt{N}}z_0(\Y) + O_G(N^{-1})
\end{equation*}
can be thought of as a first-order approximation to this log-likelihood, when $w$ is chosen to be the effective score $z_0 \circ P_0^{-1}$. For the location and shift families $\mathcal P_{\mu} = \{p_0(\cdot -\theta), \theta\in\R\}$ and $\mathcal P_{\sigma} = \{ \theta^{-1} p_0(\cdot/\theta), \theta \in \R_+ \}$, the weights $w(x) = x, w(x) = \Phi^{-1}(x)$ and $w(x) = \Phi^{-2}(x)$ are precisely of this form corresponding to logistic and Gaussian $p_0$, respectively. Moreover, given that the sample's randomness is largely captured in this first-order approximation ($\partial_{\theta} z_0(\Y)/N$ converges to a deterministic limit, a multiple of the Fisher information), it is not surprising that $T_{k,n}^w$ is fully efficient in these situations (that is; has Pitman efficiency $1$ compared to tests based on $\log \lambda_{\mathcal P}$). This analogy between $T_{k,n}^w$ and the log-likelihood of $\mathcal P$ naturally raises three questions:
\begin{enumerate}[label=(\alph*)]
    \item \textbf{Can this analogy be extended to the setting of multivariate $\Theta$?} When $\Theta \subset \mathbb R^m$ is multivariate, then so is $z_0$, and the meaning of $T_{k,n}^{z_0\circ P_0^{-1}}$ becomes less clear. $G_n$ can now approach $P_0$ from a variety of directions $\theta$, for which the likelihood locally has the form 
    \begin{equation}
        \left\langle \theta, z_0(\Y) \right\rangle / \sqrt{N} - \left\langle \theta, \partial_{\theta}z_0(\Y) \theta\right\rangle / N,
        \label{eq:theta_extension}
    \end{equation}
    and so one may expect that $z_0\circ P_0^{-1}\circ H_{k,n}(\Y)$ still yields a suitable non-parametric equivalent to $\lambda_{\mathcal P}$ in this situation. What are the appropriate analogues of the score test, Wald test, and likelihood-ratio tests, and to which extent is a non-parametric equivalent of second-order term in \eqref{eq:theta_extension} (which in the parametric setting is required to keep the log-likelihood bounded locally) needed?
    \item \textbf{Can this analogy be extended to the setting of multivariate $X$ and $Y$?} As mentioned previously, classical notions of multivariate ranks $R$ do not typically mirror most of the desirable properties of univariate ranks. The linear rank statistics discussed above rely on the probability integral transform, or CDF trick; that is, the uniformity of $P_0(X)$ whenever $X\sim P_0$ is univariate, combined with $P_0^{-1}(A)\sim X$ when $A$ is uniform on $[0,1]$. Recent work in \autocite{hallin2020distribution} characterized one appropriate extension of the CDF trick to $\R^m$ by associating the rank map $R$ with the optimal transport map transferring $P_0$ to a suitable reference measure. With such transport map in hand, a natural candidate for a multivariate extension of $T_{k,n}^w$ is $z_0\circ R^{-1}\circ R_{k,n}(\Y)$, where $R_{k,n}$ presents an appropriate empirical version of $R$. However, as mentioned previously, the cost of obtaining such $R_{k,n}$ can be prohibitive, and so identifying alternative transport maps of similar properties (relevant to two-sample testing) yet lighter computational burden is desirable. We identify such alternative transport maps, and thus obtain a non-parametric two-sample test that is both statistically and computationally efficient.
    \item \textbf{Does the non-parametric nature of $T_{k,n}^w$ afford any additional flexibility beyond $\lambda_{\mathcal P}$?} When $\mathcal P$ is misspecified, then tests based on $\lambda_{\mathcal P}$ are in general not well-calibrated \autocite{white1982maximum}. The null distribution of $T_{k,n}^w$, however, solely depends on $w$ independent of the precise shape of $F$ ($=G$ under $\Hnull$), and thus possibly allows querying more complex hypotheses. Here we explore two circumstances in which this added flexibility may be of use:
    \begin{enumerate}[label=(\arabic*)]
        \item \textit{Testing in the presence of nuisance alternatives.} Often a practitioner might like to devise a test that has power against a parametric family $\mathcal P$, while simultaneously ignoring any changes along a second model $\mathcal Q$. That is, $\mathcal Q$ acts in some sense as a nuisance model, and may, e.g., represent known measurement inaccuracies. For instance, in the context of analyzing microscopic images, differences in the actual depicted samples are meaningful to detect, whereas rotations and translations resulting from sample placement are not. We will see that $T_{k,n}^w$ admits a natural choice of $w$ that provides large power against $\mathcal P$, while remaining insensitive to $\mathcal Q$.
        \item \textit{Combining multiple weights.} Characterizing individual choices of weights that perform well against particular, narrowly defined alternatives is usually feasible; e.g., Mann-Whitney's and Klotz' tests are two such instances in the setting of location and scale shifts. A common approach to querying composite hypotheses is to apply multiple such tests in succession, and then correct the resulting multiple testing burden by any of a number of multiple testing corrections \autocite[see, e.g., ][]{bonferroni1936teoria, benjamini1995controlling, benjamini2010simultaneous, barber2015controlling, noble2009does}. Depending on the correlation structure of the tests, such correction procedures may lack power, or control false discovery rates rather than false rejection rates. We will show that $T_{k,n}^w$ provides a convenient framework for combining multiple individual weights $w_j$ in a manner that does not necessitate multiple testing corrections, while delivering competitive power in each individual direction $w_j$.
    \end{enumerate}
\end{enumerate}
The following chapters address each of these questions in turn, with additional sections discussing an extension to $K$-sample testing and exact finite sample results. Software implementing the most general $K$-sample test (and illustrating various of its optimality properties) can be found at \url{https://web.stanford.edu/~erdpham/research.html}.

\section*{Extension to multivariate parameter spaces}

As indicated above, one might hope that $z_0\circ P_0^{-1}\circ H_{k,n}$ behaves like the score function under $P_0$ and $G_n = P_{\theta/\sqrt{N}}$, even when $\mathcal P$ is multivariate. This is indeed true.
\begin{theorem}
For $k,n\to\infty, k/N\to\alpha>0$, and $w\in \R^m$ continuous, square-integrable on $(0,1)$ and $\int w = 0$, the statistic $T_{k,n}^w = w\circ H_{k,n}(\Y)/\sqrt{N}$ is asymptotically normal under local contiguous alternatives with
\begin{align*}
    &T_{k,n}^w \stackrel{\Hnull}{\longrightarrow} \mathcal N\left(0, \alpha(1-\alpha) H_w\right)
    &T_{k,n}^w \stackrel{\Hone}{\longrightarrow} \mathcal N\left(\alpha(1-\alpha) H_{w,z_0\circ P_0^{-1}}\theta, \alpha(1-\alpha) H_w\right),
\end{align*}
where $H_{w_1,w_2} = \int w_1\otimes w_2$, $H_w = H_{w,w}$, and $\Hone: X\sim p_0, Y\sim p_{\theta/\sqrt{N}}$, given that $\mathcal P$ is quadratic mean differentiable, $\int \| z_0 \circ P_0^{-1} \|_2^2$ finite, and $H_{z_0\circ P_0^{-1}}$ exists and is full-rank.
\label{thm:nonparametric_score}
\end{theorem}
\begin{proof}
    Set $\mathcal A_N = \{ w(jN^{-1}), j = 1, ..., N \}$, then under $\Hnull$, $\{ w \circ H_{k,n}(Y_1), ..., w \circ H_{k,n}(Y_n) \}$ can be obtained by sampling without replacement $n$ times from $\mathcal A_N$. Thus, the $\Hnull$ convergence of $T_{k,n}^w$ is a consequence of the multivariate central limit theorem for simple random sampling \autocite[see, e.g., ][where the univariate version is stated; the multivariate generalization is a direct consequence of applying the Cram{\'e}r-Wold device]{rosen1964limit}, after verifying that
    \begin{equation*}
        \max_{a \in \mathcal A_N} \frac{\| H^{-1/2}_w a \|^2}{N} \to 0
    \end{equation*}
    by the continuity and square-integrability assumptions on $w$. The specific covariance structure is calculated to be
    \begin{equation*}
        \Var T_{k,n}^w = \frac{1}{N} \sum_{j=1}^n \E w(r_j) \otimes w(r_j) + \frac{1}{N} \sum_{m\neq j} \E w(r_m) \otimes w(r_j) \xrightarrow{\Hnull} (1-\alpha) H_w - (1-\alpha)^2 H_w = \alpha(1-\alpha) H_w,
    \end{equation*}
    as desired.
    \\[2mm]
    The $\Hone$ results is shown by way of LeCam's third lemma. Writing the log-likelihood-ratio statistic
    \begin{equation*}
        \rho_n = \sum_{j=1}^n \log \frac{ p_{\theta/\sqrt{N}}(Y_j) }{ p_0(Y_j) }
    \end{equation*}
    in its local expansion
    \begin{equation*}
        \rho_n = \left\langle \frac{ z_0(\Y) }{ \sqrt{N} }, \theta \right\rangle - \frac{1}{2} \left\langle \theta, H_{z_0\circ P_0^{-1}}\theta \right\rangle + o_{\Hnull}(1) = \left\langle \frac{ z_0 \circ H_{k,n}^{-1} \circ H_{k,n}(\Y) }{ \sqrt{N} }, \theta \right\rangle - \frac{1}{2} \left\langle \theta, H_{z_0\circ P_0^{-1}}\theta \right\rangle + o_{\Hnull}(1),
    \end{equation*}
    and augmenting $\mathcal A_N$ to
    \begin{equation*}
        \mathcal A_N^+ = \left\{ w\left( \frac{j}{N} \right) + e_{m+1} \otimes \left\langle z_0 \circ H_{k,n}^{-1} \left(\frac{j}{N}\right), \theta \right\rangle, j = 1, ..., N \right\} = \left\{ w\left( \frac{j}{N} \right) + e_{m+1} \otimes \left\langle z_0(Z_j), \theta \right\rangle \right\}
    \end{equation*}
    yields the joint vector $(T_{k,n}^w, \rho_n)$ as a size-$n$ sample without replacement from $\mathcal A_N^+$ up to $o_{\Hnull}(1)$ and a deterministic shift. Though $\mathcal A_N^+$ depends on $\Z$ (and therefore is random),
    \begin{equation*}
        \P_{\Hnull}\left[ \bigg\{ N^{-1} \sum_{a \in \mathcal A_N^+} a \to 0 \bigg\} \cap \bigg\{ N^{-1} \sum_{a\in \mathcal A_N^+ a} a^2 \to \Sigma \bigg\} \right] = 1,
    \end{equation*}
    for some to be determined covariance matrix $\Sigma \in \R^{(m+1) \times (m+1)}$, and so the central limit theorem for simple random sampling applies to show the joint normality of $(T_{k,n}^w, \rho_n)$ under $\Hnull$. By LeCam's third lemma, the shift under $\Hone$ is then given by
    \begin{equation*}
        \E T_{k,n}^w \rho_n = \E\left(\frac{1}{N} \sum_{j=1}^n w(r_j) \otimes z_0\circ H_{k,n}^{-1}(r_j) + \frac{1}{N} \sum_{m\neq j} w(r_m) \otimes z_0 \circ H_{k,n}^{-1}(r_j)\right) \theta = (A_1 + A_2)\theta.
    \end{equation*}
    By the law of large numbers for sampling without replacement \autocite[][section 5]{rosen1964limit}, $A_1$ converges to $(1-\alpha)H_{w, z_0\circ P_0^{-1}}$, while 
    \begin{equation*}
        A_2 = \frac{1-\alpha}{n} \E \sum_{m\neq j} \E\left[ w(r_m) \otimes z_0 \circ H_{k,n}^{-1}(r_j) \mid \Z \right] = -\frac{(n-1)(1-\alpha)}{N(N-1)} \E \sum_{s=1}^N w(sN^{-1}) \otimes z_0 \circ H_{k,n}^{-1}(sN^{-1}),
    \end{equation*}
    which converges to $(1-\alpha)^2 H_{w, z_0\circ P_0^{-1}}$. Consequently
    \begin{equation*}
        \E T_{k,n}^w \rho_n \xrightarrow{\Hnull} (1-\alpha)H_{w, z_0 \circ P_0^{-1}} - (1-\alpha)^2 H_{w, z_0\circ P_0^{-1}} = \alpha(1-\alpha) H_{w, z_0 \circ P_0^{-1}},
    \end{equation*}
    and the theorem is proved.
\end{proof}
Theorem \ref{thm:nonparametric_score} allows for hypothesis testing as long as $k$ and $n$ are large enough; for small to moderate sample sizes, see the section on exact finite sample distributions. Tests based on $T_{k,n}^w$ and $T_{k,n}^{\mu + \sigma w}$, with $\mu\in\R^m, \sigma\in\R$ are equivalent since $(T_{k,n}^{\mu + \sigma w} - n\mu)/\sigma = T_{k,n}^w$, and so the centering $\int w = 0$ is purely to ease notation and does not impact the generality of Theorem \ref{thm:nonparametric_score}. This centering will be assumed implicitly in all following results, as will be the regularity assumptions on $\mathcal P$ (unless stronger conditions are needed), which guarantee local asymptotic normality of $\rho_n$. Big-$N$ results like these are sufficient to compare tests in the Pitman sense.
\begin{corollary}
The Pitman efficiency of $T_{k,n}^{z_0\circ P_0^{-1}}$ relative to the Neyman-Pearson test is $\alpha$ for local alternatives $p_{\theta/\sqrt{N}}\in\mathcal P$ converging to $p_0$.
\end{corollary}
\begin{proof}
    By local asymptotic normality, the log-likelihood ratio $\rho_n$ converges to
    \begin{align*}
        &\rho_n \xrightarrow{\Hnull} \mathcal N\left( -\frac{\sqrt{1-\alpha}}{2} H_{z_0\circ P_0^{-1}}\theta, H_{z_0\circ P_0^{-1}} \right),
        &\rho_n \xrightarrow{\Hone} \mathcal N\left( \frac{\sqrt{1-\alpha}}{2} H_{z_0\circ P_0^{-1}}\theta, H_{z_0\circ P_0^{-1}} \right).
    \end{align*}
    The corollary then follows from the $(\mu_1'/\mu_2')(\sigma_2/\sigma_1)$ characterization of Pitman efficiency \autocite[Corollary 15.2.1 of ][]{lehmann2005testing}.
\end{proof}
The rate of $\alpha$ is, of course, a result of not knowing $F$ exactly. The fairer comparison is against the generalized likelihood ratio test.
\begin{corollary}
The Pitman efficiency of $T_{k,n}^{z_0\circ P_0^{-1}}$ against $-2\log \lambda_{\mathcal P}$ is $1$ for alternatives $p_{\theta/\sqrt{N}}\in\mathcal P$ converging to $p_0$, assuming that the maximum-likelihood estimates involved in $\lambda_{\mathcal P}(\X, \Y)$ are efficient in the sense of \autocite[][eq. (14.62)]{lehmann2005testing}.
\label{cor:pitman_likelihood_ratio}
\end{corollary}
\begin{proof}
    The regularity assumptions on $\mathcal P$ guarantee that $-2 \log \lambda_{\mathcal P}(\X, \Y)$ converges to
    \begin{align*}
        &-2\log \lambda_{\mathcal P}(\X, \Y) \xrightarrow{\Hnull} \chi^2_m,
        &-2\log \lambda_{\mathcal P}(\X, \Y) \xrightarrow{\Hone} 
        \left\| \mu + Z_m \right\|_2^2,
    \end{align*}
    where 
    \begin{equation*}
        \mu = \left[\Pi_{(\sqrt{\alpha},\sqrt{1-\alpha})}^{\perp} \otimes 1_m\right]
        \left[\begin{pmatrix}
            \sqrt{\alpha} & 0 \\
            0 & \sqrt{1-\alpha}
        \end{pmatrix}
        \otimes H_{z_0\circ P_0^{-1}}^{1/2}\right] \left[ e_2 \otimes \theta \right],
    \end{equation*}
    with $1_m$ the $m\times m$ identity matrix and $\Pi_{x}^{\perp}$ the projection onto $(\operatorname{span} x)^{\perp}\subset \R^2$ for any $x\in\R^2$. Consequently,
    \begin{equation*}
        \mu =
        \begin{pmatrix}
            -\sqrt{\alpha}(1-\alpha) \\
            \alpha\sqrt{1-\alpha}
        \end{pmatrix}
        \otimes H_{z_0\circ P_0^{-1}}^{1/2} \theta,
    \end{equation*}
    and the result follows as
    \begin{equation*}
        \| \mu \|_2^2 = \alpha(1-\alpha) \| H_{z_0\circ P_0^{-1}}^{1/2} \theta \|_2^2. \qedhere
    \end{equation*}
\end{proof}
The efficiency assumption on the maximum likelihood estimates allow, in essence, to employ Taylor expansion as if $\mathcal P$ was twice continuously differentiable, and prove corresponding normality results \autocite[][Theorem 14.4.1]{lehmann2005testing}; and will be implicit in all future statements. In order to compute the weight function in Theorem \ref{thm:nonparametric_score} and Corollary \ref{cor:pitman_likelihood_ratio}, a candidate $\theta^* \in \Theta$ which is believed to govern $\X$ needs to be known. In practice, this is not always the case, and misspecified weights may not be fully efficient. This suggests a strategy of adaptively determining $w$.
\begin{theorem}
Set $\hat\theta = \arg\max_{\theta} p_{\theta}(\Z)$, $\hat w = z_{\hat\theta}\circ P_{\hat\theta}^{-1}$, and assume $\hat\theta \to 0$ almost surely, then
\begin{enumerate}[label=(\alph*)]
    \item $(T_{k,n}^{\hat w} \mid \hat \theta)$ is exactly distribution-free under $\Hnull$.
    \item Asymptotically,
        \begin{align*}
            &T_{k,n}^{\hat w} \stackrel{\Hnull}{\longrightarrow} \mathcal N\left(0, \alpha(1-\alpha) H_{z_0\circ P_0^{-1}}\right)
            &T_{k,n}^{\hat w} \stackrel{\Hone}{\longrightarrow} \mathcal N\left(\alpha(1-\alpha) H_{z_0\circ P_0^{-1}}\theta, \alpha(1-\alpha) H_{z_0\circ P_{0}^{-1}}\right),
        \end{align*}
        with notation as in Theorem \ref{thm:nonparametric_score}.
    \item The Pitman efficiency of $T_{k,n}^{\hat w}$ compared to $-2\log \lambda_{\mathcal P}$ against local alternatives $p_{\theta^* + \theta/\sqrt{N}}\in\mathcal P$ is $1$ for any $\theta^*\in\Theta$.
\end{enumerate}
\label{thm:adaptive}
\end{theorem}
\begin{proof}
    (c) follows from (b) and the same computation as in Corollary \ref{cor:pitman_likelihood_ratio}, and so only (a) and (b) remain to be shown.
    \begin{enumerate}[label=(\alph*)]
        \item Denoting by $\E^{X}$ the projection onto the sigma-algebra $\sigma(X)$ generated by $X$, one has for any continuous, bounded $f$ on $\R^m$
        \begin{multline*}
            \E^{\hat \theta}_{\Hnull} f(T_{k,n}^{\hat w}) = \E^{\hat \theta}_{\Hnull} \E^{\Z}_{\Hnull} f(T_{k,n}^{\hat w}) = \E^{\hat \theta}_{\Hnull} \left[ \frac{1}{\binom{N-1}{k-1}} \sum_{0 < m_1 < ... < m_n \leq N} f( \sum_{j=1}^n \hat w(m_j/N) ) \right]
            \\
            = \frac{1}{\binom{N-1}{k-1}} \sum_{0 < m_1 < ... < m_n \leq N} f( \sum_{j=1}^n \hat w(m_j/N) ),
        \end{multline*}
        since $\hat \theta$ is $\sigma(\Z)$-measurable, and $(\{ r_1, ..., r_n \} \mid \Z)$ is uniform over the $n$-subsets of $\{1, ..., N\}$ under $\Hnull$. The right-hand side is independent of $p_0$ (except through $\hat w$), therefore proving the claim.
        \item Setting $\mathcal A_N = \{ \hat w(jN^{-1}), j = 1,..., N\}$ and
        \begin{equation*}
            \mathcal A_N^+ = \left\{ \hat w\left( \frac{j}{N} \right) + e_{m+1} \otimes \left\langle z_0 \circ H_{k,n}^{-1} \left(\frac{j}{N}\right), \theta \right\rangle, j = 1, ..., N \right\},
        \end{equation*}
        the almost sure convergence of $\hat \theta$ combined with the continuous mapping theorem imply that
        \begin{equation*}
            \P \left[ \E_{\mathcal A_N} X \to 0, \Var_{\mathcal A_n} A_N \to \Sigma, \E_{\mathcal A_N^+} X \to 0, \Var_{\mathcal A_N^+} \to \Sigma^+ \right] = 1,
        \end{equation*}
        for some fixed choices of $\Sigma$ and $\Sigma^+$, and so the same reasoning as employed in Theorem \ref{thm:nonparametric_score} carries through. \qedhere
    \end{enumerate}
\end{proof}
Thus, similar to $T_{k,n}^{\Phi^{-1}}$ being preferable to the $t$-test at least asymptotically in the Pitman sense, $T_{k,n}^{\hat w}$ recovers performance of any likelihood-ratio test in the same sense, at least for alternatives within $\mathcal P$, while remaining non-parametric. Moreover, unlike the generalized likelihood-ratio test, $T_{k,n}^{\hat w}$ behaves well for alternatives falling outside $\mathcal P$.
\begin{theorem}
Given a second model $\mathcal Q$ with $\mathcal Q \setminus \mathcal P$ non-empty, local alternatives $q_{\gamma/\sqrt{N}}$ converging to $q_0 \in \mathcal Q \setminus \mathcal P$ (that is, $X\sim q_0, Y\sim q_{\gamma/\sqrt{N}}$), and almost sure convergence of $\hat\theta$ to $\theta^* = \arg\max_{\theta} \E_{q_0} \log p_{\theta}(Y)$, asymptotically
\begin{align*}
    &T_{k,n}^{\hat w} \stackrel{\Hnull}{\longrightarrow} \mathcal N\left(0, \alpha(1-\alpha) H_{z_{\theta^*} \circ P_{\theta^*}^{-1}} \right)
    &T_{k,n}^{\hat w} \stackrel{\textnormal{H}_{1,\mathcal Q}}{\longrightarrow} \mathcal N\left(\alpha(1-\alpha) K\gamma, \alpha(1-\alpha) H_{z_{\theta^*}\circ P_{\theta^*}^{-1}}\right),
\end{align*}
where $K = \int z_{\theta^*} \circ P_{\theta^*}^{-1} \otimes s_0(Y) \circ Q_0^{-1}$, for $s_0$ the score function associated with $\mathcal Q$ at $\gamma = 0$.
\label{thm:projection_model}
\end{theorem}
\begin{proof}
    The convergence proof under $\Hnull$ proceeds as in part (b) of Theorem \ref{thm:adaptive}. The shift under $\Hone$, is again given by
    \begin{equation*}
        \E T_{k,n}^{\hat w}\rho_n = \left( \frac{1-\alpha}{n} \E \sum_{j=1}^n \hat w(r_j) \otimes q_0 \circ H_{k,n}^{-1} (r_j) + \frac{1-\alpha}{n} \E \sum_{m\neq j} \hat w(r_m) \otimes s_0 \circ H_{k,n}^{-1}(r_m) \right)\gamma = (A_1 + A_2)\gamma,
    \end{equation*}
    with the same arguments as in Theorem \ref{thm:nonparametric_score} showing
    \begin{align*}
        &A_1 \xrightarrow{\Hnull} (1-\alpha) H_{z_{\theta^*} \circ P_{\theta^*}^{-1}, s_0 \circ Q_0}
        &A_2 \xrightarrow{\Hnull} (1-\alpha)^2 H_{z_{\theta^*} \circ P_{\theta^*}^{-1}, s_0 \circ Q_0},
    \end{align*}
    as desired.
\end{proof}
\begin{remark}
For the location and shift models the respective effective score functions $z_{\mu}\circ (P^{-1}+\mu) = -p'/p$ and $z_{\sigma}\circ (\theta P^{-1})(x) = -x p'(x)/p(x)$ are independent of the parameters, and so $T_{k,n}^{\hat w} = T_{k,n}^{z}$. In these cases, the distribution of the likelihood-ratio test under alternatives can be computed explicitly, and the efficiencies obtained from Theorem \ref{thm:projection_model} reduce to the well-known efficiencies considered in, e.g., \autocite{hodges1956efficiency} and \autocite{klotz1962nonparametric}.
\end{remark}
\begin{remark}
An important computational advantage of likelihood-ratio tests is their associated Wilks' phenomenon; that is, their limiting $\chi^2_m$-distribution does not depend on any parameters of the model, and in particular does not require estimation of the Fisher information $\E_{P_0} z_0\otimes z_0$. As $T_{k,n}^w$ is anchored in a Taylor expansion of the log-likelihood, dispensing with the Fisher information altogether is difficult, but of course a similar limiting $\chi^2_m$ distribution can be obtained from Theorems \ref{thm:nonparametric_score} and \ref{thm:adaptive}. For $M_{k,n}^w = \max_{\theta} \langle \theta, T_{k,n}^w\rangle - \| \theta \|_H$, asymptotically
\begin{align*}
    &M_{k,n}^{w} \stackrel{H_0}{\longrightarrow} \chi^2_{m}
    &M_{k,n}^{w} \stackrel{H_{1,\mathcal P}}{\longrightarrow} \left \| Z_m + \alpha(1-\alpha) H^{1/2}_{w,z_0\circ P_0^{-1}}\theta \right\|_2^2,
\end{align*}
where $Z_m\in \R^m$ is a standard normal variable. The same results hold true for $w$ replaced by $\hat w$.
\label{rmk:chi_square}
\end{remark}
Although the preceding analysis focused on Pitman efficiencies, and therefore local alternatives, it is useful to clarify the behavior of $T_{k,n}^{\hat w}$ under fixed $F,G$. As $T_{k,n}^w$ constitutes a linear approximation to the likelihood, it is not surprising that the usual non-linear KL-projection condition is replaced by a linear one. The following is essentially a restatement of Theorem 3.2 in \autocite{deb2021efficiency}, and is included for completeness.
\begin{theorem}[\autocite{deb2021efficiency}]
For $F\neq G$ fixed, tests based on $T_{k,n}^w$ are consistent as long as $\E_G w\circ (\alpha F + (1-\alpha)G ) \neq 0$. That is, under $X\sim F, Y\sim G$, $\mathbb P\left( \textnormal{reject }\Hnull \right) \to 1$.
\label{thm:consistency}
\end{theorem}
\begin{remark}
For the Mann-Whitney test, $w(x) = x-1/2$, and so the condition $\E_g w\circ (\alpha F + (1-\alpha)G )\neq 0$ amounts to $0 \neq \E_g F - 1/2 = \mathbb P(X<Y) -1/2$, recovering a well-known result that follows from the original work in \autocite{mann1947test}, but was pointed out explicitly by \autocite{lehmann1951consistency}.
\end{remark}

\section*{Extension to multivariate sample spaces}

Part of the appeal of likelihood-based methods is their straightforward extension to the multivariate setting. Once a model is specified, be it on $\R$ or $\R^p$, quantities like $\lambda_{\mathcal P}$ behave similarly, at least asymptotically, with their limiting distribution solely depending on $m$ rather than $p$. Such phenomena invite similar considerations for $T_{k,n}^w$. The previous section's analysis relied heavily on the empirical CDF $H_{k,n}$ of $\Z$, whose relevant properties are most directly derived from the canonical ordering of $\R$. Recent work by \autocite{hallin2020distribution} revealed that it is, in fact, not the canonical ordering that gives rise to such desired properties, but rather that they may be recovered through a certain minimization task to which the rank map constitutes a solution. This observation ties ranks to the field of optimal transport \autocite[see, e.g., ][]{peyre2019computational}, thereby enabling their generalization to the multivariate setting; see the work of \autocite{hallin2020distribution} itself, as well as the broad exposition given in \autocite{deb2021efficiency} for how precisely it does so. The univariate behavior of $T_{k,n}^w$ readily extends to this situation.
\begin{theorem}
For $w$ square-integrable and continuously differentiable on $[0,1]^p$, $H_{k,n}$ the optimal transport map from $\Z\subset \R^p$ to $\mathcal C_N^p\subset [0,1]^p$, where $\mathcal C_N^p$ is any set of points whose associated empirical point measure weakly to the uniform law on $[0,1]^p$ and $\max_{c\in \mathcal C_N^p} N^{-1}\| c \|_2^2 \to 0$, define $T_{k,n}^w(\X, \Y) = w\circ H_{k,n}(\Y) / \sqrt{N}$. Then the statements of Theorems \ref{thm:nonparametric_score}-\ref{thm:consistency} remain true with the domain of integration defining $H = \int w\otimes w$ replaced by $[0,1]^p$, and CDFs interpreted as optimal transport maps from their respective distributions to $\mathrm{Uniform}\left( [0,1]^p \right)$.
\label{thm:multivariate_sample_space}
\end{theorem}
\begin{proof}
    Statements about asymptotic distributions under $\Hnull$ relied solely on two properties:
    \begin{enumerate}[label=(\alph*)]
        \item $\{ H_{k,n}(Y_1), ..., H_{k,n}(Y_n) \} = \{ r_1, ..., r_n \}$ is uniformly distributed over all $n$-subsets of $\{jN^{-1}, j=1, ..., N\}$; or equivalently, it is a sample without replacement of size $n$ from $\{jN^{-1}, j=1,..., N \}$.
        \item $\{jN^{-1}, j=1,..., N\} = \{ H_{k,n}(Z_1), ..., H_{k,n}(Z_N)\}$ converges weakly to $\mathrm{Uniform}([0,1])$.
    \end{enumerate}
    When interpreted as an optimal transport map between two point clouds in $\R^p$, $H_{k,n}$ still satisfies both (a) and (b) with $\{jN^{-1}, j=1, ..., N\}$ and $[0,1]$ replaced by $\mathcal C_N^p$ and $[0,1]^p$, respectively, and so proofs in Theorems \ref{thm:nonparametric_score}-\ref{thm:consistency} concerning $\Hnull$ may be reused essentially verbatim. To see that convergence under local contiguous alternatives remains unchanged, it suffices to observe that under $\Hnull$ $(T_{k,n}^w, \rho_n)$ can still be written as a simple random sample from $\mathcal A_N^+$ by property (a), and $\mathrm{Cov}(T_{k,n}^w, \rho_n)$ is identical to its univariate counterpart with $[0,1]$ replaced by $[0,1]^p$ by part (b).
\end{proof}
\begin{example*}[Gaussian location model]
For the Gaussian location model $\mathcal P = \{ \mathcal N(\theta, 1_{m}) : \theta\in\R^m \}$ and $\hat w = \Phi^{-1}$ (applied component-wise), $T_{k,n}^{\hat w}(\X, \Y)$ is exactly equivalent to the rank Hotelling-$T^2$ test with Gaussian effective reference distribution, as proposed in \autocite{deb2021efficiency}. In particular, it exhibits the same favorable Pitman efficiencies compared to Hotelling's $T^2$ against a large class of alternatives.
\label{ex:gauss_location}
\end{example*}
Theorem \ref{thm:multivariate_sample_space} provides theoretical insight, and is useful for sample sizes $N$ that are small enough for optimal transport maps to be tractably computed, yet large enough for asymptotic statements to become relevant. As discussed previously, these constraints limit practical applications. The following theorem softens these constraints by accelerating computation.
\begin{theorem}
    For $w$ as in Theorem \ref{thm:multivariate_sample_space}, $N=a^p$ and $\mathcal C_N^p = \{j\cdot a^{-1}, j=0, ..., a-1\}^p$, define the $j$-th component $H_{k,n}^j(z)$ of $H_{k,n}(z)$ (for $z$ in $\Z$) iteratively through $H_{k,n}^j(z) = \lfloor r_j(z)/a^{p-j} \rfloor \cdot a^{-1}$, where $r_j(z)$ is one less than the rank of $\langle e_j, z \rangle$ among $\{\langle e_j, s\rangle, s\in \Z: r_{j-1}(s) = r_{j-1}(z)\}$ and $r_0 \equiv 0$. Then Theorem \ref{thm:multivariate_sample_space} holds with this choice of $H_{k,n}$, and the entire two-sample test can be carried out in $O(pN\log N)$.
    \label{thm:conditional_map}
\end{theorem}
\begin{proof}
    As in the proof of Theorem \ref{thm:multivariate_sample_space}, $H_{k,n}$ satisfies both properties (a) and (b) and thus the conclusion under $\Hnull$ follows. Similarly, the result under $\Hone$ follows by an identical computation of $\operatorname{Cov}(T_{k,n}^w, \rho_n)$ and noting that $H_{k,n}$ converges (under $\Hnull$) to the (invertible) population map
    \begin{equation}
        P_0: (z_1, ..., z_p) \mapsto \left(\P_0\left[Z_{\ell} \leq z_{\ell} \mid Z_j = z_j, j=1, ..., \ell-1\right],\ell=1, ..., p \right),
        \label{eq:population_conditional}
    \end{equation}
    where $\P_0$ is the measure governing $\Z$. The computational complexity is obtained by observing that the $j$-th layer requires sorting $a^{p+1-j}$ elements $a^{j-1}$ times, resulting in $O\left(\sum_{j=1}^p a^{j-1}\cdot a^{p+1-j}\log a^{p+1-j}\right) = O\left( pN \log N \right)$ calculations.
\end{proof}
\begin{remark}
    The assumption of $N$ being a perfect power of $p$ is simply to streamline the proof of Theorem \ref{thm:conditional_map} and communicate its key ingredients more clearly. It can be relaxed to arbitrary $N$ in a straightforward manner: instead of $\mathcal C_n^p$ being a product space of ranks, it assigns $\approx [N^{1/p}]$ along each marginal direction, balancing points as evenly as possible. $H_{k,n}$ then performs the same ranking procedure as before. This general testing algorithm is implemented in the code provided with this paper.
\end{remark}
\begin{remark}
    When an underlying parametric family $\mathcal P$ is assumed and efficiency desired, the weight computation/estimation requires inversion of the population map \eqref{eq:population_conditional}, just like efficiency in the context of Theorem \ref{thm:multivariate_sample_space} requires inversion of a (population) optimal transport map. While optimal transport maps are typically difficult to invert (indeed, closed-form formulae are rare even for the forward map), the explicit expression in \eqref{eq:population_conditional} should facilitate such inversion. Moreover, even when such explicit inversion is intractable, the low computational complexity of $H_{k,n}$ enables approximations through, e.g., simulations.
\end{remark}

\section*{Considerations beyond efficiency}

Although the framework presented above is fully efficient in the case of correct model specification, the more appealing feature is arguably its robustness against model misspecification. This robustness can be exploited for applications that might otherwise be challenging in a likelihood context. This section details two such instances.

\subsection*{Projecting out nuisance alternatives}

The outcomes of scientific experiments are typically influenced by two sources of randomness: intrinsic fluctuations of the phenomenon of interest (e.g., particle paths behaving approximately like an ensemble of Brownian motions) and measurement error (e.g., variations in photon densities reaching the microscope, diffraction, plate contamination, etc.); these two sources will be referred to as signal and noise. While differences in signal are usually meaningful for understanding the underlying phenomenon at hand, variation in noise is merely reflective of inconsistencies in the measurement process itself, and so is ideally discarded. That is, given a sample $\X$ drawn from some $p = p(p_{\text{signal}}, p_{\text{noise}})$ whose details depend on signal and noise distributions $p_{\text{signal}}, p_{\text{noise}}$, respectively, it is of interest to decide whether the generating mechanism of a second sample $\Y$ differs from that of $\X$ in the signal component $p_{\text{signal}}$, whereas changes in the noise contribution $p_{\text{noise}}$ are to be disregarded. A special case of such situation is given by models of the form $\mathcal P = \{ p_{\theta_0, \theta_1} : \theta_j \in \Theta_j \}$ for some parameter spaces $\Theta_0,\Theta_1$, where the signal varies in $\Theta_0$ with the noise parametrized by $\Theta_1$. Writing $\X \stackrel{\textnormal{iid}}{\sim} p_{\theta_0,\theta_1}, \Y\stackrel{\textnormal{iid}}{\sim} p_{\vartheta_0, \vartheta_1}$, the task then becomes to decide whether $\theta_0 = \vartheta_0$ against $\vartheta_0 = \theta_0 + a/\sqrt{N}$.

\begin{example*}[Mixture location model]
\par{In genomics it is often of interest to detect \textit{differential expression} of genes; that is, whether the distribution of a certain gene attribute (say, expression) varies across conditions. One approach to do so consists of collecting tissue samples across the conditions of interest, aggregating gene expression within each condition, and then comparing the resulting measurements via two-sample tests. Such procedure, however, is confounded by tissue composition; that is, distinct tissue samples may be comprised of differing proportions of cell types that constitute the tissue. Generally, genes exhibit differential expression across cell types, and consequently, most two-sample tests will detect shifts in proportions even when expression profiles themselves remain unaltered \autocite{lowe2014correcting}. So-called cell-type deconvolution methods provide one attempt to overcome such confounding by inferring cell-type proportions from the data prior to any differential expression analysis \autocite[see, e.g., the discussion in][]{avila2020benchmarking}. The majority of such methods, though, rely on cell-type specific gene expression reference panels, and thereby implicitly assume the absence of any differential expression. How they interact with differential expression is generally unclear. Moreover, distributional properties of the inferred proportions are (with few exceptions \autocite[see, e.g., ][]{erdmann2021likelihood}) not provided, and thus the precise manner in which proportion estimates ought to enter two-sample tests remains ambiguous.}
\par{As long as the number of cells contained in each tissue sample is sufficiently large, the central limit theorem can be invoked to describe the distribution of $g$ genes' expression $X\in\R^g$ in a single tissue sample as $\mathcal N(M\rho, W_{(M\odot M)\rho + \Sigma\rho})$, where $\rho \in \Delta^{c-1}$ are the proportions of each of $c>1$ cell types, $W_x$ is the diagonal matrix with vector $x$ on its diagonal, $\odot$ represents the entry-wise product, and $M, \Sigma \in \R^{g\times c}$ record the expectations and variances of every gene's expression profile in each cell type \autocite{erdmann2021likelihood}. When tissue samples share similar proportions and behave independently and identically under a given condition, one is thus interested in testing whether the signal $M$ differs across conditions, while ignoring the noise $\rho$.}
\label{ex:differential_expression}
\end{example*}

Theorem \ref{thm:nonparametric_score} shows that local alternatives with associated score function $z$ induce a shift in $T_{k,n}^w$ proportional to $\mu_w(z) = H_{w,z\circ P^{-1}} \theta = \int w \langle z\circ P^{-1}, \theta \rangle$. Consequently, maximizing power along a signal $z_0$, while guaranteeing robustness against a nuisance model $z_1$ reduces to solving a linear program in $\mathcal K = L^2[(0,1)^p]^m$.

\begin{theorem}
For $\hat \theta = \arg\max_{\theta} p_{\theta}(\Z)$, and 
\begin{align*}
    &\begin{pmatrix}
        \hat w_0 \\ \hat w_1
    \end{pmatrix}
    = z_{\hat \theta}\circ P_{\hat \theta}^{-1}
    &\begin{pmatrix}
        w_0 \\ w_1
    \end{pmatrix}
    = z_{0}\circ P_{0}^{-1},
\end{align*}
define the effective weight $\hat w = \hat w_0 - \hat H_{w_0, w_1} \hat H_{w_1}^{-1} \hat w_1$ with $\hat H_{w_0,w_1}$ and $\hat H_{w_1}$ the sub-matrices of the usual estimates of the Fisher information associated with $\mathcal P$ at $\theta = 0$. Then as long as $\hat \theta \to 0$ almost surely,
\begin{align*}
    &T_{k,n}^{\hat w} \xrightarrow{\Hnull} \mathcal N(0, \alpha(1-\alpha) (H_0 - H_{0,1}H_1^{-1}H_{1,0})) 
    \\
    &T_{k,n}^{\hat w} \xrightarrow{\Hone} \mathcal N(\alpha(1-\alpha) (H_0 - H_{0,1}H_1^{-1}H_{1,0})(e_1 \otimes 1_m)\theta, \alpha(1-\alpha) (H_0 - H_{0,1}H_1^{-1}H_{1,0})),
\end{align*}
where $H_0 = H_{w_0}, H_1 = H_{w_1}, H_{j,k} = H_{w_j, w_k}$. In particular, $\E_{\Hone} T_{k,n}^{\hat w} \to 0$ whenever $\theta = (0, \theta_1) \in \Theta_0 \times \Theta_1$ for any $\theta_1$.
\label{thm:signal_noise}
\end{theorem}
\begin{proof}
    By Theorem \ref{thm:adaptive},
    \begin{equation*}
        T_{k,n}^{z_{\hat\theta} \circ P_{\hat \theta}^{-1}} \xrightarrow{\Hnull} \mathcal N\left( 0, \alpha(1-\alpha) H_{z_0\circ P_0^{-1}} \right).
    \end{equation*}
    Observing that $T_{k,n}^{\hat w} = \hat\Pi T_{k,n}^{z_{\hat \theta}\circ P_{\hat \theta}^{-1}}$, where $\hat \Pi = e_1\otimes 1_m - e_2 \otimes \hat H_{0,1}\hat H_1^{-1} \rightarrow \Pi = e_1 \otimes 1_m - e_2 \otimes H_{0,1}H_1^{-1}$ by the continuous mapping theorem, an application of Slutsky's theorem shows that
    \begin{equation*}
        T_{k,n}^{\hat w} = \hat \Pi T_{k,n}^{z_{\hat \theta}\circ P_{\hat \theta}^{-1}} \xrightarrow{\Hnull} \mathcal N\left(    0, \alpha(1-\alpha) \Pi H_{z_0 \circ P_0^{-1}} \Pi^* \right).
    \end{equation*}
    A quick computation verifies that $\Pi H_{z_0\circ P_0^{-1}} \Pi^* = H_0 - H_{0,1}H_1^{-1}H_{1,0}$ as desired. A similar application of Slutsky's theorem and quick computation of $\Pi H_{z_0 \circ P_0^{-1}}$ establishes the convergence result under $\Hone$.
\end{proof}
Of course, the very same principle may be employed even in the absence of explicit models: For any two weights $w_0, w_1$, the effective weight $w = w_0 - H_{w_0,w_1} H_{w_1}^{-1} w_1$ will maximize power in the direction of $w_0$ and orthogonal to $w_1$.
\begin{example*}[Mixture location model continued]
To characterize when differential expression can be locally robustly distinguished from changes in proportion, it suffices by Theorem \ref{thm:signal_noise} to compute the weight functions $w_{\rho}\oplus w_M = z_{\rho,M}$ corresponding to $\rho$ and $M$, respectively, and check whether $w_M$ is not in the span of the components of $w_{\rho}$. Taking $p_{\rho,M}^{\Sigma}(x)$ to be the density of a $\mathcal N(M\rho, W_{(M\odot M)\rho + \Sigma\rho})$ variable, some calculations yield
\begin{align*}
    &w_{\rho} = \sum_{\gamma \in [g]} \frac{M_{\gamma}}{\sigma(\rho, M_{\gamma}, \Sigma_{\gamma})} \Phi^{-1}\circ \Pi_{\gamma} + \frac{M_{\gamma}^2 + \Sigma_{\gamma}^2 - 2\mu(\rho, M_{\gamma})M_{\gamma} }{2\sigma^2(\rho, M_{\gamma}, \Sigma_{\gamma})} \left( \frac{\Phi^{-2}\circ \Pi_{\gamma}}{2} - 1 \right) 
    \\
    &w_{M} = \sum_{\gamma \in [g]} e_{\gamma} \otimes \left[ \frac{\rho}{\sigma(\rho, M_{\gamma}, \Sigma_{\gamma})} \Phi^{-1}\circ \Pi_{\gamma} + \frac{\rho \odot M_{\gamma} - \mu(\rho, M_{\gamma})\rho }{\sigma^2(\rho, M_{\gamma}, \Sigma_{\gamma})} \left( \frac{\Phi^{-2}\circ \Pi_{\gamma}}{2} - 1 \right) \right],
\end{align*}
where $M_{\gamma}, \Sigma_{\gamma}\in \R^c$ are the expectations and variances of gene $\gamma$ across the $c$ cell types, $\mu(\rho, M_{\gamma}), \sigma^2(\rho, M_{\gamma}, \Sigma_{\gamma})$ indicate the expectation and variance of the corresponding $\rho$-mixture, and $\Pi_{\gamma}: \R^g \to \R$ is the $\gamma^{\text{th}}$ coordinate projection. The components of both $w_{\rho}$ and $w_M$ are linear combinations of the $2g$ linearly independent functions $\Phi^{-1}\circ \Pi_1, ..., \Phi^{-1}\circ \Pi_g, \Phi^{-2}\circ \Pi_1, ..., \Phi^{-2}\circ \Pi_g$, and themselves (generically) linearly independent. Therefore, isolating shifts in $M$ is possible as long as $c<2g$; in particular, signal and noise can be locally distinguished even when $g+1<c<2g$, in which case the deconvolution task itself is not identifiable.
\end{example*}
\begin{remark}
Theorem \ref{thm:signal_noise} allows for selectively ignoring local alternatives; discarding global alternatives is, in general, more challenging. Indeed, for most models the linear closure of $\{ z_{\theta_1,\theta_2} : \theta_2 \in \Theta_2 \}$ will have infinite rank or span all of $\mathcal K$.
\end{remark}

\subsection*{Qualitative and multiple testing}

The early design of weight functions for linear rank statistics did not draw from the connections to likelihood functions as emphasized here, but rather attempted to qualitatively anticipate the impact of distinct shapes on the resulting test statistic. For instance, weights that are increasing and skew-symmetric around $x=1/2$ should naturally be sensitive to shifts in location, while convex and symmetric shapes may detect perturbations in scale; \autocite[see, e.g., the discussions in][]{van1953neuer,ansari1960rank}. Many such tests do not naturally correspond to any single model \autocite{mood1954asymptotic, siegel1960nonparametric}, yet have proven useful in various applications. The machinery developed above can typically be incorporated into this qualitative approach to two-sample testing in a straightforward manner. E.g., the same considerations around monotonicity and skew-symmetry, and convexity and evenness for detecting location and scale shifts, respectively, can be applied in order to yield corresponding two-sample tests acting on multivariate samples. Indeed, the choice of $w = J$ for some score function $J:\R^p\to \R^p$ recovers exactly the \textit{rank Hotelling} $T^2$ tests discussed in \autocite{deb2021efficiency}, if $J$ is the inverse transport map from $[0,1]^p$ to an "effective" reference distribution $P$. $P$ with independent components allow univariate reasoning to be transferred directly to the individual components; for example, the increasing and skew-symmetric components of distribution functions associated with laws like the uniform measure on $[0,1]^p$ itself or $\mathcal N(0,1_{p})$ provide suitable generalizations of the Mann-Whitney-$U$ test and its Gaussian-score-transformed version (cf. the Example on page \pageref{ex:gauss_location}); while $w_1(x) = |x - 1/2|, w_2(x) = (x - 1/2)^2$ and $w_3(x) = \Phi^{-2}(x)$ (all applied component-wise) naturally extend the tests of Siegel-Tukey, Mood, and Klotz to the multivariate task of sensing scale shifts. Theorem \ref{thm:conditional_map} indicates that power properties of such generalized tests should behave similarly to the those of the univariate tests they generalize, at least for large classes of null distributions.
\begin{example*}[Relative efficiencies of multivariate scale tests]
Under local alternatives $p_{\theta}(x) = \theta^{-1}p_0(\theta^{-1}x)$ around $\theta=1$ ($x\in \R^p$), the tests of Siegel-Tukey, Mood, and Klotz behave asymptotically as $\| Z + \sqrt{\alpha(1-\alpha)} \beta_j \|_2^2$ (cf. Remark \ref{rmk:chi_square}), where $Z$  is a univariate standard Gaussian variable, and
\begin{align*}
    \beta_1 &= -\sqrt{\frac{48}{p}} \int_{[0,1]^p} \left(\| x-1/2 \|_1 - \frac{p}{4} \right) \left\langle P_0^{-1}(x), z_0\circ P_0^{-1}(x) \right\rangle 
    \\
    \beta_2 &= -\sqrt{\frac{180}{p}} \int_{[0,1]^p} \left(\| x-1/2 \|_2^2 - \frac{p}{12} \right) \left\langle P_0^{-1}(x), z_0\circ P_0^{-1}(x) \right\rangle
    \\
    \beta_3 &= -\sqrt{\frac{2}{p}} \int_{[0,1]^p} \left(\| \Phi^{-1}(x) \|_2^2 - p \right) \left\langle P_0^{-1}(x), z_0\circ P_0^{-1}(x) \right\rangle,
\end{align*}
with the usual interpretation of $P_0$ as a transport map pushing $p_0$ to $\mathrm{Uniform}([0,1]^p)$. From the form of the $\beta_j$, it is clear that the multivariate relative efficiencies between $w_1, w_2$ and $w_3$ reduce to their univariate analogues \autocite[reviewed, and partially worked out, in][]{klotz1962nonparametric} when, e.g., $P_0$ has independent components. Moreover, if $P_0$ is elliptically symmetric (that is, $p_0(x) = Z^{-1} \exp\circ g(\| \Sigma^{-1/2}x\|_2^2)$ for some sufficiently regular function $g$, covariance matrix $\Sigma$, and normalization constant $Z$), the same conclusion holds; for in such case,
\begin{equation*}
    \frac{1}{2}\left\langle P_0^{-1}(x), z_0\circ P_0^{-1}(x) \right\rangle = g'(\|\Sigma^{-1/2}P_0^{-1}(x)\|_2^2) \left\| \Sigma^{-1/2}P_0^{-1}(x) \right\|_2^2,
\end{equation*}
and, writing $f(x) = g'(\|\Sigma^{-1/2}P_0^{-1}(x)\|_2^2)\| \Sigma^{-1/2} P_0^{-1}(x) \|_2^2$,
\begin{equation*}
    \int \sum_{j=1}^p (c\circ \Pi_j) f = p \int (c\circ \Pi_1) f
\end{equation*}
for $j\in \{1,2,3\}$ and any function $c$ due to the spherical symmetry of $f\left(\mathrm{Uniform}([0,1]^p)\right)$.
\end{example*}
In addition to lacking a precise generative model and therefore necessitating the design of more qualitative weights as discussed above, many applications exhibit multiple plausible directions in which a system may shift under the alternative. E.g., in the differential expression task considered in page \pageref{ex:differential_expression}'s Example, it typically is not clear a priori whether a certain gene's abundance should be expected to vary in location across distinct conditions, or, say, spread \autocite[which has been shown to assume important roles in modulating gene networks; see, e.g.,][and references therein]{eldar2010functional}. Testing for each such direction separately and correcting the resulting $p$-values for their multiple testing burden risks under-powering of the corresponding hypothesis test. As long as a suitable weight can be formulated for each query of interest, Theorem \ref{thm:nonparametric_score} allows for effective joint testing.
\begin{example*}[Testing for location and scale shifts simultaneously]
For samples $\X,\Y \in \R^p$, pick $w_{\mu} = \Phi^{-1}, w_{\sigma} = \Phi^{-2}$, and form the combined weight $w(x) = (w_{\mu}, w_{\sigma})\in \R^{2p}$. Under $\Hnull$, the resulting $T_{k,n}^w$ statistic is asymptotically distributed as $\mathcal N(0,\alpha(1-\alpha) H)$, where $H = p (e_1+2e_2)\otimes 1_{p})$, and so can be used to perform hypothesis testing with. Under $p_{\theta/\sqrt{N}}(x) = | W_{\sigma^{-1}_N} | p_0(W_{\sigma^{-1}_N} x - \mu_N)$, where by slight abuse of notation $\mu_N = \mu/\sqrt{N}, \sigma_N^{-1} = (\sigma_1^{-1}, ..., \sigma_p^{-1})/\sqrt{N}$ for some $\mu, \sigma\in \R^p$, and $W_a$ is the diagonal matrix featuring $a$ on its diagonal as before, $\| \sqrt{\alpha(1-\alpha)}H^{-1/2} T_{k,n}^w \|_2^2$ is asymptotically distributed like 
\begin{equation}
    \left\| Z_{2p} + \sqrt{\alpha (1-\alpha)} \left[\int (\Phi^{-1}, \Phi^{-2}) \otimes \left(z_0\circ P_0^{-1}(x), W_x [z_0\circ P_0^{-1}(x)] \right) \right]^{1/2} \left( e_1 \otimes \mu + e_2 \otimes \sigma \right) \right\|_2^2.
    \label{eq:multiple_testing_comparison}
\end{equation}
Assuming for simplicity that $p_0$ is the density of a standard normal variable in $\R^{p}$, \eqref{eq:multiple_testing_comparison} simplifies to 
\begin{equation*}
    \left\| Z_{2p} + \sqrt{p \alpha(1-\alpha)}(e_1 \otimes \mu + \sqrt{2}e_2 \otimes \sigma) \right\|_2^2.
\end{equation*}
The goal is to compare the power resulting from such expression to that of, say, a Bonferroni correction applied to $T_{k,n}^{w_{\mu}}$ and $T_{k,n}^{w_{\sigma}}$ separately. This is easiest computed in the large-$p$ regime, for then
\begin{multline}
    \P\left[\frac{p^{-1/2}}{2}\left(\left\| Z_{2p} + \sqrt{p \alpha(1-\alpha)}(e_1 \otimes \mu + \sqrt{2}e_2 \otimes \sigma) \right\|_2^2 - 2p\right) > \frac{p^{-1/2}}{2}\left(\chi^2_{2p,1-\alpha}-2p\right)\right] 
    \\
    \approx 1 - \Phi(z_{1-\alpha} - \| \tilde{\mu} \|_2^2 - 2 \| \tilde{\sigma} \|_2^2),
    \label{eq:multiple_testing_joint}
\end{multline}
where $\chi^2_{2p,1-\alpha}$ and $z_{1-\alpha}$ are the $(1-\alpha)^{\text{th}}$ quantiles of a $\chi^2_{2p}$ and a standard Gaussian variable, respectively, and 
\[(\| \tilde{\mu} \|_2^2, \| \tilde{\sigma} \|_2^2) = \left[2\alpha(1-\alpha)\right]^{-1} \lim_{p\to\infty} \sqrt{p} (\| \mu \|_2^2, \| \sigma \|_2^2). \]

\begin{figure}[h!]
    \centering
    \includegraphics[width=\textwidth]{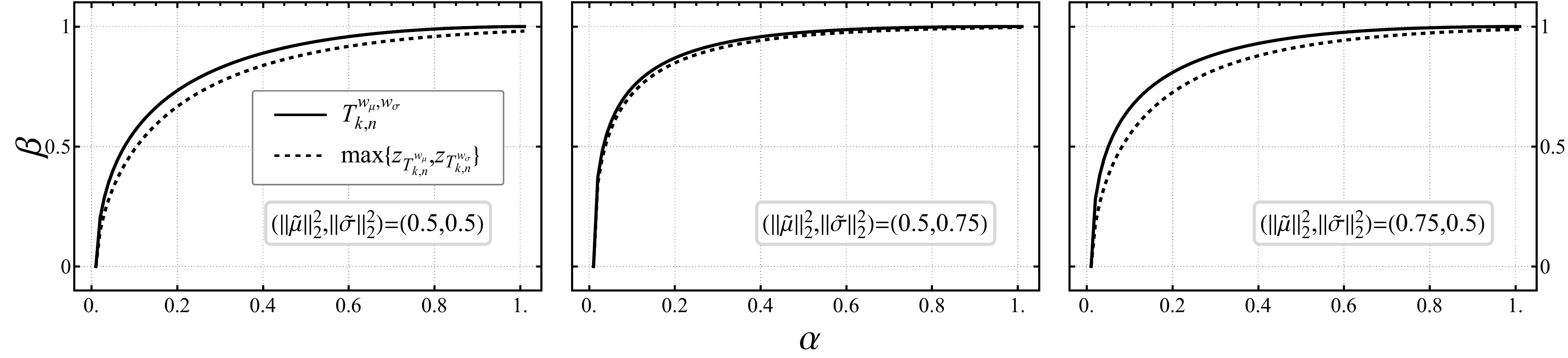}
    \caption{Comparison of simultaneously testing for location and scale shifts either jointly ($T_{k,n}^{w_{\mu},w_{\sigma}}$, cf. \eqref{eq:multiple_testing_joint}) or by correcting for multiple testing ($\max\{ z_{T_{k,n}^{w_{\mu}}}, z_{T_{k,n}^{w_{\sigma}}} \}$, cf. \eqref{eq:multiple_testing_correction}).}
    \label{fig:multiple_testing_comparison_figure}
\end{figure}

Writing
\begin{align*}
    E_1 &= \left\{ (2p)^{-1/2}\left(\left\| Z_{p} + \sqrt{p \alpha(1-\alpha)}\mu \right\|_2^2 - p\right) > (2p)^{-1/2}\left(\chi^2_{p,1-\alpha/2}-p\right) \right\}
    \\
    E_2 &= \left\{ (2p)^{-1/2}\left(\left\| Z_{p} + \sqrt{2p \alpha(1-\alpha)}\sigma \right\|_2^2 - p\right) > (2p)^{-1/2}\left(\chi^2_{p,1-\alpha/2}-p\right) \right\}
\end{align*}
The corresponding quantity after correcting for multiple testing of $T_{k,n}^{w_{\mu}}$ and $T_{k,n}^{w_{\sigma}}$ reads
\begin{equation}
    \P\left[E_1 \cup E_2\right] 
    \approx 1 - \Phi(z_{1-\alpha/2} - \sqrt{2} \| \tilde{\mu} \|_2^2)\cdot \Phi(z_{1-\alpha/2} - 2\sqrt{2} \| \tilde{\sigma} \|_2^2).
    \label{eq:multiple_testing_correction}
\end{equation}
For fixed choices of $\| \tilde{\mu} \|_2^2, \| \tilde{\sigma} \|_2^2$, \eqref{eq:multiple_testing_joint} and \eqref{eq:multiple_testing_correction} can be compared as functions of the nominal size $\alpha$. Figure \ref{fig:multiple_testing_comparison_figure} does so for three representative pairs $(\| \tilde{\mu} \|_2^2, \| \tilde{\sigma} \|_2^2) \in \{ (0.5,0.5), (0.5,0.75), (0.75,0.5) \}$, illustrating a noticeable gain in power.
\end{example*}

\subsection*{General $K$-sample testing}

The power of tests based on $T_{k,n}^w$ is invariant under translations and scaling of $w$ (see discussion after Theorem \ref{thm:nonparametric_score}), and so one may without loss of generality assume that $\sum_{j=1}^N w(N^{-1}j) = 0$. Under this convention, $(1-\alpha_{k,n})\int w\circ H_{k,n} \ \mathrm{d}G_n = -\alpha_{k,n} \int w\circ H_{k,n} \ \mathrm{d}F_k$, and so the roles of $\X$ and $\Y$ are entirely symmetric. That is, even though previous results were stated in terms of distributional perturbations of $\Y$, they equally apply to those of $\X$. This ceases to remain true once more than two samples are considered, in which case a more general strategy is to be employed. More concretely, given $K$ samples $\mathcal X_1, ..., \mathcal X_K$ of sizes $n = (n_1, ..., n_K)$, the $K$-sample testing task consists of deciding if \textit{all} samples were generated from identical distributions, and has found applications that rival those of two-sample tests in breadth \autocite[see, e.g., discussion in][]{scholz1987k, thas2010comparing}. If an underlying model $\mathcal P$ from which the $\mathcal X_j$ are generated independently is known, then
\begin{equation*}
    \lambda_{\mathcal P}\left( \mathcal X_1, ..., \mathcal X_K \right) = \frac{\max_{\theta_1, ..., \theta_K \in \Theta} \prod_{j=1}^K p_{\theta_j}(\mathcal X_j)}{\max_{\theta \in \Theta} \prod_{j=1}^K p_{\theta}(\mathcal X_j)}
\end{equation*}
is a natural extension of \eqref{eq:likelihood_ratio}, and like \eqref{eq:likelihood_ratio} features an associated Wilks phenomenon: under the null hypothesis of $\mathcal X_1, ..., \mathcal X_K \stackrel{\text{iid}}{\sim} p_0$, $-2 \log \lambda_{\mathcal P}(\mathcal X_1, ..., \mathcal X_K)$ behaves asymptotically like a $\chi^2_{(K-1)m}$ variable. If the null distribution is known to be uniform, then under local perturbations of $(\mathcal X_1, ..., \mathcal X_K)$ by $\theta/\sqrt{N} = (\theta_1, ..., \theta_K)/\sqrt{N} \in \R^{mK}$ (with $N = n_1 + ... + n_K$) a first-order expansion of $-2 \log \lambda_{\mathcal P}(\mathcal X_1, ..., \mathcal X_K)$ reads
\begin{equation*}
    -2\log \lambda_{\mathcal P}(\mathcal X_1, ..., \mathcal X_K) \approx \max_{\theta_1, ..., \theta_K} \frac{1}{\sqrt{N}} \sum_{k=1}^K \left\langle z_0(\mathcal X_k), \theta_k \right\rangle = \max_{\theta_1, ..., \theta_K} \frac{1}{\sqrt{N}} \left(\sum_{k=1}^K e_k \otimes z_0(\mathcal X_k) \right)\theta,
\end{equation*}
motivating the following $K$-sample version $\tilde T_{n}^w$ of $T_{k,n}^w$:
\begin{equation*}
    \tilde T_n^w = \sqrt{N}\sum_{k=1}^K e_k \otimes \left( \alpha_k \int w \circ H_n \ \mathrm{d} F_k \right) = \frac{1}{\sqrt{N}} \sum_{k=1}^K e_k \otimes w\left( \mathcal X_k \right) \in \R^{mK},
\end{equation*}
where $H_n$ is a transport map of $\mathcal Z_n = \mathcal X_1 \cup ... \cup \mathcal X_K$ to $\mathcal C^p_N$, $\alpha_k = n_k/N$, and $F_k$ the corresponding map for $\mathcal X_k$. Given the dependence of $\tilde T_n^w$'s components ($\tilde T_n^w \in \ker{ \mathbbm{1}_K \otimes 1_m } = \Theta_0$), it is more convenient to work with a transformed version $T_n^w = \Gamma \tilde T_n^w \in \R^{(m-1)K}$, where $\Gamma = B^* \lambda_{1/\sqrt{\alpha}} \otimes 1_m: \left( \R^K \otimes \R^m \right) \to \left( \R^{K-1} \otimes \R^m \right)$, and $B$ an orthonormal basis for $\Theta_1 = \Theta_0^{\perp}$. As is the case with $T_{k,n}^w$, the large-$N$ behavior of $T_n^w$ can be worked out.
\begin{theorem}
For $\min_k n_k\to\infty, n_k/N\to\alpha_k>0$, and the usual conditions on $w\in \R^m$ and $\mathcal P$, the statistic $\sqrt{N} T_{n}^w$ is asymptotically normal under the null and local (contiguous) alternatives with
\begin{align*}
    &T_{n}^w \stackrel{\Hnull}{\longrightarrow} \mathcal N\left(0, B^* A A^* B \otimes H_w\right)
    &T_{n}^w \stackrel{\Hone}{\longrightarrow} \mathcal N\left(B^*AA^* \lambda_{\sqrt{\alpha}} \otimes H_{w,z_0\circ P_0^{-1}}\theta, B^* A A^* B\otimes H_w\right),
\end{align*}
where $A$ is an orthonormal basis for $\ker \sqrt{\alpha} \otimes \sqrt{\alpha}$, with $\sqrt{\cdot}$ applied component-wise.
\label{thm:k_samples}
\end{theorem}
\begin{proof}
    To simplify notation, the proof only treats the case $p=1$, from which $p>1$ follows in a straightforward manner by reasoning as in Theorem \ref{thm:multivariate_sample_space}. As usual, the $\Hnull$ result is a consequence of the central limit theorem for sampling without replacement. More concretely, write $T_n^w(k) = (e_k \otimes 1_m) \tilde T_{n}^w$ and set $\mathcal A_N(k+1) = \{ w(jN^{-1}), j=1,..., N\} \setminus \cup_{m=1}^k \mathcal B_m$, where $\mathcal B_m = \{ w\circ H_n(X_{m,1}), ..., w\circ H_n(X_{m,n_m})\}$. Then $\tilde T_n^w(1)$ converges to a centered normal variable by Theorem \ref{thm:nonparametric_score}, and
    \begin{equation*}
        \tilde T_n^w(k+1) \ \bigg| \ \sum_{j=1}^k \tilde T_n^w(j) = s
    \end{equation*}
    is a sample without replacement of size $n_{k+1}$ from
    \begin{equation*}
        \mathcal A_N(k+1) \ \bigg| \ N^{-1/2}\sum_{b\in\cup_{j=1}^k\mathcal B_j} b = s
    \end{equation*}
    and thus is amenable to a central limit theorem too. Therefore, $(\tilde T_n^w(1), ..., \tilde T_n^w(K-1))$ is jointly normal as $N$ grows large, implying that
    \begin{equation*}
        \tilde T_n^w
        = 
        \begin{pmatrix}
            \tilde T_n^w(1) \\ \vdots \\ \tilde T_n^w(K-1) \\ \tilde T_n^w(K)
        \end{pmatrix}
        =
        \begin{pmatrix}
            \tilde T_n^w(1) \\ \vdots \\ \tilde T_n^w(K-1) \\ 1 - \sum_{k=1}^{K-1} \tilde T_n^w(k)
        \end{pmatrix}
        \xrightarrow{\Hnull} \mathcal N\left( 0, \tilde \Sigma \right)
    \end{equation*}
    with $\ker \tilde \Sigma = \Theta_0$. Consequently, $T_n^w \to \mathcal N(0, \Sigma)$, where $\Sigma$ is given by
    \begin{multline*}
        \Var T_n^w = \E \sum_{k=1}^K \frac{1}{n_k} (B^*e_k \otimes B^* e_k) \otimes (w\circ H_{k,n}(\mathcal X_k) \otimes w \circ H_{k,n}(\mathcal X_k)) 
        \\
        + \E \sum_{m\neq k} \frac{1}{\sqrt{n_m n_k}} (B^*e_m \otimes B^* e_k) \otimes (w\circ H_{k,n}(\mathcal X_m) \otimes w\circ H_{k,n}(\mathcal X_k)) = A_1 + A_2.
    \end{multline*}
    Computations similar to those used in the proof of Theorem \ref{thm:nonparametric_score} yield
    \begin{align*}
        &A_1 = \sum_{k=1}^K (1-\alpha_k) (B^* e_k \otimes B^* e_k) \otimes H_w
        &A_2 = \sum_{m\neq k} \sqrt{\alpha_m \alpha_k} (B^* e_m \otimes B^* e_k) \otimes H_w,
    \end{align*}
    and thus
    \begin{equation*}
        \Var T_n^w = B^*(1_K - \sqrt{\alpha} \otimes \sqrt{\alpha}) B \otimes H_w = B^* A A^* B \otimes H_w,
    \end{equation*}
    as advertised.
    \\[2mm]
    Establishing the $\Hone$ convergence can be performed along similar lines as before: augment $\mathcal A_N(k+1)$ to
    \begin{equation*}
        \mathcal A_N^+(k+1) = \{ w(jN^{-1}) + e_{m+1} \otimes z_0 \circ H_n^{-1} (jN^{-1}) \} \setminus \cup_{m=1}^k \mathcal B_m^+
    \end{equation*}
    where
    \begin{equation*}
        \mathcal B_j^+ = \{ w\circ H_n(X_{j,1}) + e_{m+1} \otimes z_0(X_{j,1}), ..., w\circ H_n(X_{j,n_j}) + e_{m+1} \otimes z_0 (X_{j,n_j}) \},
    \end{equation*}
    then the same argument showing the joint normality of $(\tilde T_n^w(1), ..., \tilde T_n^w(K))$ can be used to arrive at the joint normality of $(\tilde T_n^w(1), ..., \tilde T_n^w(K), z_0(\mathcal X_1), ..., z_0(\mathcal X_K))$, and therefore the joint normality of $(T_n^w, \rho_n)$, where
    \begin{equation*}
        \rho_n = \sum_{k=1}^K \sum_{j=1}^{n_k} \log \frac{p_{\theta_k/\sqrt{N}}(X_{k,j})}{ p_0(X_{k,j}) } = \sum_{k=1}^K \sqrt{\alpha_k} \left\langle z_{0}(\mathcal X_k), \theta_k \right\rangle - \frac{1}{2} \left\langle \lambda_{\alpha} \otimes H_{z_0 \circ P_0^{-1}} \theta, \theta \right\rangle + o_{\Hnull}(1)
    \end{equation*}
    is the relevant likelihood ratio statistic. The usual appeal to LeCam's third lemma together with the observation that
    \begin{equation*}
        \E T_n^w \rho_n = \sum_{k=1}^K \frac{\sqrt{\alpha_k}}{n_k} B^* e_k \otimes \E\left[ w\circ H_n(\mathcal X_k) \otimes z_0(\mathcal X_k) \right]\theta_k + \sum_{m\neq k} \frac{\sqrt{\alpha_m}}{\sqrt{n_m n_k}} B^* e_k \otimes \E \left[ w\circ H_n (\mathcal X_k) \otimes z_0(\mathcal X_m) \right]\theta_m
    \end{equation*}
    converges to
    \begin{multline*}
        \sum_{k=1}^K \sqrt{\alpha_k}(1-\alpha_k) B^*e_k \otimes H_{w,z_0 \circ P_0^{-1}}\theta_k - \sum_{m\neq k} \alpha_m \sqrt{\alpha_k} B^* e_k \otimes H_{w,z_0\circ P_0^{-1}} \theta_j 
        \\ 
        = \left[ B^*\left( \lambda_{\sqrt{\alpha}}) - \sqrt{\alpha} \otimes \alpha  \right) \otimes H_{w,z_0\circ P_0^{-1}} \right]\theta = B^* AA^* \lambda_{\sqrt{\alpha}} \otimes H_{w,z_0\circ P_0^{-1}}
    \end{multline*}
    completes the proof.
\end{proof}
As before, the Pitman efficiency of $T_n^w$ compared to $\lambda_{\mathcal P}(\mathcal X_1, ..., \mathcal X_K)$ is favorable.
\begin{corollary}
The Pitman efficiency of $T_{n}^{z_{0}\circ P_0^{-1}}$ relative to a test based on $\lambda_{\mathcal P}(\mathcal X_1, ..., \mathcal X_K)$ is $1$ for local alternatives $p_{\theta/\sqrt{N}}\in\mathcal P$ converging to $p_{0}$.
\label{cor:k_sample_efficiency}
\end{corollary}
\begin{proof}
    The usual regularity assumptions on $\mathcal P$ imply that
    \begin{equation}
        -2 \log \lambda_{\mathcal P}(\mathcal X_1, ..., \mathcal X_K) \xrightarrow{\Hone} \left\| Z_{(m-1)K} + \Pi_{\lambda_{\sqrt{\alpha}} \otimes H_{z_0\circ P_0^{-1}} \Theta_0} \lambda_{\sqrt{\alpha}} \otimes H_{z_0\circ P_0^{-1}}^{1/2} \theta \right\|_2^2,
        \label{eq:k_sample_ratio_shift}
    \end{equation}
    where $\Pi_{A}: \Theta^K \to A$ for a linear subspace $A \subset \Theta^K$ is the orthogonal projection onto $A$. Here, $\Theta_0$ is the diagonal of $\Theta^K$, and so
    \begin{equation*}
        \lambda_{\sqrt{\alpha}}\otimes H_{z_0\circ P_0^{-1}}^{1/2}\Theta_0 = \ker \left( 1_K - \sqrt{\alpha} \otimes \sqrt{\alpha} \right) \otimes 1_m,
    \end{equation*}
    giving for the non-centrality parameter in \eqref{eq:k_sample_ratio_shift} $AA^*\lambda_{\sqrt{\alpha}} \otimes H_{z_0\circ P_0^{-1}}^{1/2} \theta$ as desired.
\end{proof}
Likewise, results mirroring those contained in Theorems \ref{thm:adaptive}-\ref{thm:consistency} can be obtained in this general $K$-sample setting. In particular, Corollary \ref{cor:k_sample_efficiency} remains true if $w = z_{\hat \theta}$ is chosen adaptively through maximum-likelihood estimation based on $\mathcal Z_n$.

\subsection*{Exact finite sample distributions}
\label{sec:exact_finite_sample}

Although typical Berry-Esseen rates of $O(N^{-1/2})$ are known for linear rank statistics under mild regularity conditions on $w$ \autocite{friedrich1989berry}, modern scientific applications often must accommodate to very small sample sizes, where asymptotic results may not be of relevance yet \autocite[see, e.g.,][]{mollan2020precise}. Empirical simulation studies have shown that linear rank statistics continue to perform favorably in such regime \autocite{conover1981comparative, conover2018update}, and so characterizing the finite sample distribution of $T_{k,n}^w$ is therefore desirable both theoretically as well as in practice. While an explicit description appears to be currently out of reach, \autocite{erdmann2022generalized} recently showed that such characterization is available for a closely related, asymptotically equivalent, statistic $S_{k,n}^w$ based on rank-spacings (i.e., $r_{j+1}-r_j$) instead of ranks, and argued that its small-sample power properties ought to be comparable to those of $T_{k,n}^w$ for most purposes. More concretely, $S_{k,n}^w$ is computed as
\begin{equation*}
    S_{k,n}^w = \frac{1}{n} \sum_{j=0}^n w\left( \frac{j}{n} \right)(r_{j+1}-r_j) = w(1) -\int H_N(x) \left( w \circ G_n(x) - w\circ (G_n - n^{-1})(x) \right)  \ \mathrm{d}G_n,
\end{equation*}
with the convention $r_{0} = 0, r_{n+1}=n+1$, and where $\X, \Y$ and $w$ are all assumed univariate \autocite[][provides a slightly different definition, though the two versions differ merely by a constant]{erdmann2022generalized}. A first step towards adapting $S_{k,n}^w$ for our purposes is to allow $w$ be in $\R^m$ rather than just $\R$.
\begin{theorem}
    For $\X, \Y \subset \R, w: [0,1] \to \R^m$, the Laplace transform of $S_{k,n}^w$ is given by
    \begin{multline}
        \mathbb E_{\Hnull} e^{\langle s, S_{k,n}^w \rangle} = (k-1)(-1)^{n+1}\cdot e^{t k w_{\max}(s)} \times \sum_{j=1}^{n-1} a_j^{e^{t(\langle e_j, Ws\rangle -w_{\max}(s))}} \bigg[ b_{k,n} \left( 1 - e^{t(\langle e_j, Ws\rangle - w_{\max}(s))(n+k-1)}\right)        
        \\
        + \sum_{m=0}^{n-3} c_{k,n,m} \left( 1 - e^{t(\langle e_j, Ws\rangle - w_{\max}(s))} \right)^{n-2-m} \bigg],
        \label{eq:finite_sample_multivariate}
    \end{multline}
    for almost every $s \in \R^m$, where $W\in \R^{n\times m}$ has rows $w(jN^{-1}), j=1,..., n$, $w_{\max}(s) = \| Ws \|_{\max}$, for any $r\in\R^{n}$, $a_j^{r} = \prod_{m\neq j} (r_j-r_m)^{-1}$ and  
    \begin{align*}
        &b_{k,n} = \frac{(-1)^{n}}{n+k-1}\cdot\binom{n+k-2}{n-2}^{-1} &c_{k,n,m} = \frac{(-1)^m}{k+1}\cdot \frac{\binom{k-2}{m}}{\binom{n+m+1}{m}},
    \end{align*}
    as long as $\ker W = \{0\}$.
    \label{thm:multivariate_transform}
\end{theorem}
\begin{proof}
    The proof is a direct consequence of Theorem 1 in \autocite{erdmann2022generalized} after observing that $\langle s, S_{k,n}^w \rangle = S_{k,n}^{\langle s, w\rangle}$. The condition on $\ker W$ guarantees distinctness of the components of $Ws$.
\end{proof}
\begin{remark}
    Similar to the expression obtained in \autocite{erdmann2022generalized}, \eqref{eq:finite_sample_multivariate} may appear inconvenient to work with analytically, but can often be inverted numerically.
\end{remark}
\begin{remark}
    The condition on $\ker W$ can be relaxed at the expense of slightly more complicated expressions, see discussion around Theorem 1 of the univariate case.
\end{remark}
The extension to multivariate sample spaces is more involved.
\begin{theorem}
    Tile the unit cube $[0,1]^p$ by a collection $\mathcal S$ of $\lceil N^{1/p}\rceil^p$ sub-cubes, each of side-length $\lceil N^{-1/p} \rceil$ in the natural way, and let $\pi: \{1, ..., | \mathcal S | \} \to \mathcal S$ be an enumeration of $\mathcal S$ with $\pi_{j+1}$ being adjacent to $\pi_j$ (that is, sharing a face) for all $1 \leq j < | \mathcal S |$. Associate with $\pi$ a non-self-intersecting curve $c: [0,1] \to [0,1]^p$ that linearly interpolates between the centers of $\pi_j$ and $c(j/|\mathcal S|) = \mathrm{center}(\pi_j)$. Construct the point measure $\mathcal C^p_N$ as $\mathcal C^p_N = N^{-1} \sum_{j=1}^N \delta_{c(jN^{-1})}$, and optimally transport $\Z$ to $\mathcal C^p_N$ via $H_N$ with induced transport maps $F_k, G_n$ on $\X, \Y$; define $\rho_j = \inf \{ s: | c([0,s]) \cap G_n(\Y) | \geq j \}$. Then a test based on
    \begin{equation*}
        S_{k,n}^{w\circ c} = \frac{1}{n} \sum_{j=1}^n w\circ c \left( \frac{j}{n} \right) (\rho_{j+1} - \rho_j)
    \end{equation*}
    is asymptotically equivalent to one using $T_{k,n}^w$ under $\Hnull$, and the Laplace transform of $S_{k,n}^{w \circ c}$ is given by \eqref{eq:finite_sample_multivariate} (replacing $w$ with $w\circ c$ throughout).
    \label{thm:skn_multivariate_samples}
\end{theorem}
\begin{proof}
    The construction essentially reduces the multivariate case to the univariate one: as $H_{k,n}(\Y)$ is uniform over $\mathcal C^p_N$ under $\Hnull$, $(\rho_1, ..., \rho_n)$ are uniform over subsets of size $n$ from $\{1, ..., N\}$. The theorem then follows from convergence results of the univariate $S_{k,n}^w$, and Theorem \ref{thm:multivariate_transform}.
\end{proof}
\begin{remark}
    A $\pi$ that is convenient to work with in practice can be constructed by induction: for $p=2$, label each sub-cube in $\mathcal S$ by $(x_1,x_2) \in \{1, ..., |\mathcal S|^{1/2} = a_2 \}^2$ and set 
    \begin{equation*}
        \pi_{2,j} = \left(\lfloor j/a_2 \rfloor, (-1)^{\lfloor j/a_2 \rfloor}\left[ 1 + \left( j-1 \mod a_2 \right) \right] \right).
    \end{equation*}
    For $p>2$, similarly label each sub-cube by a $p$-tuple $(x_1, ..., x_p) \in \{ 1, ..., |\mathcal S|^{1/p} = a_p \}^p$, and define $\pi_{p,j} = ( \lfloor j/a_p^{p-1} \rfloor, \pi_{p-1, j - \lfloor j/a_p^{p-1} \rfloor} )$.
\end{remark}
\begin{remark}
    Convergence under local contiguous alternatives in the univariate setting has been worked out in \autocite{holst1980asymptotic}. The proofs there are established under a univariate $w$; however, it appears that this assumption can be relaxed to $n^{-1}\sum_{j=1}^n w(jn^{-1}) \to \int w$. Therefore, as long as $c$ in Theorem \ref{thm:skn_multivariate_samples} is chosen such that $n^{-1}\sum_{j=1}^n w\circ c(jn^{-1})\to \int w$ (as is, e.g., the case with the choice given in the remark above), such convergence results ought to transfer to the multivariate setting, in which case $S_{k,n}^w$ and $T_{k,n}^w$ are asymptotically equivalent even under local alternatives. Rigorously checking these conditions is left for future work.
\end{remark}

\section*{Conclusion}

While the local optimality of linear rank statistics for certain choices of coefficients has been known in several specific situations (e.g., location shifts), a general framework embedding such instances into maximum-likelihood-type arguments appears to not have been articulated thus far. By doing so explicitly, we hope to provide guidance to practitioners who may suspect certain generative models underlying a given application, but would like to guard against misspecification. Such situations arise frequently in modern science, where, e.g., the broad physical or biological (stochastic) mechanisms at play are understood, yet experimental measurements add less well characterized noise. With the complexity of hypotheses and experimental protocols increasing, so must the complexity of statistical models; in particular, building tests that can incorporate multivariate sample and parameter spaces is desirable. Building on recent progress tying together multivariate ranks and optimal transport, we identify alternative transport maps that extend linear rank statistics from the univariate to the fully multivariate setting while preserving all their favorable properties at improved computational cost. For problems whose complexity is well captured by the framework surrounding Pitman efficiency, the resulting adaptive linear multi-rank statistics always outperform likelihood-ratio tests, and thus should be preferred. Viewing linear rank statistics from both a non-parametric testing angle and a maximum-likelihood perspective opens up applications that might not be straightforward through either approach alone; here we showcased this flexibility on the examples of nuisance parameters, multiple testing and general $K$-sample testing. Most of the findings presented pertain to the large-sample regime, though finite-$N$ guarantees become available for the closely related linear rank-spacings statistics. A software package implementing the full $K$-sample test is made available with this paper.

\section*{Acknowledgments}

The author warmly thanks G{\"u}nther Walther for comments that inspired the section on projecting out nuisance parameters.

\printbibliography

\end{document}